\documentclass[10pt,twoside]{siamart1116}

\usepackage[english]{babel}
\usepackage{graphicx,epstopdf,epsfig}
\usepackage{amsfonts,epsfig,fancyhdr,graphics, hyperref,amsmath,amssymb}

\usepackage{amssymb}
\usepackage{amsmath}
\makeatletter
\def\maketag@@@#1{\hbox{\m@th\normalfont\normalsize#1}}
\makeatother
\numberwithin{equation}{section}
\usepackage{mathtools}
\usepackage{graphicx} 
\usepackage{float}  
\usepackage{geometry}
\usepackage{color}
\usepackage{graphicx}
\usepackage{xcolor}

\DeclareMathOperator{\abs}{abs}

\usepackage{tikz}
\usetikzlibrary{matrix, positioning}

\newcommand\overmat[2]{%
  \makebox[0pt][l]{$\smash{\color{black}\overbrace{\phantom{%
    \begin{matrix}#2\end{matrix}}}^{\text{\color{black}#1}}}$}#2}

\def\cvd{~\vbox{\hrule\hbox{%
     \vrule height1.3ex\hskip0.8ex\vrule}\hrule } }

\newcommand{\Names}{Moshe Goldberg, Daniel Hershkowitz, and Daniel Szyld}
\newcommand{\Title}{Maximum Absolute Determinants of Upper Hessenberg Bohemian Matrices}

\newtheorem{remark}[theorem]{Remark}

\newtheorem{conjecture}[theorem]{Conjecture}
\newtheorem{problem}[theorem]{Problem}
\newtheorem{notation}[theorem]{Notation}
\newtheorem{observation}[theorem]{Observation}

\newcommand{\reals}{\mathbb{R}}

\renewcommand{\theequation}{\thesection.\arabic{equation}}
\renewtheorem{theorem}{Theorem}[section]

\begin{document}

\bibliographystyle{plain}

\setcounter{page}{1}

\thispagestyle{empty}

 \title{\Title}

\author{Jonathan P. Keating\thanks{Mathematical Institute, University of Oxford, Andrew Wiles Building, Oxford, OX2 6GG, UK (keating@maths.ox.ac.uk)}
\and
Ahmet Abdullah Kele\c{s}\thanks{Department of Mathematics, Bilkent University, 06800 Bilkent, Ankara, Turkey (abdullah.keles@ug.bilkent.edu.tr)} }

\markboth{\Names}{\Title}

\maketitle

\begin{abstract}
A matrix is called Bohemian if its entries are sampled from a finite set of integers. We determine the maximum absolute determinant of upper Hessenberg Bohemian Matrices for which the subdiagonal entries are fixed to be $1$ and upper triangular entries are sampled from $\{0,1,\dots,n\}$, extending previous results for $n=1$ and $n=2$ and proving a recent conjecture of Fasi \& Negri Porzio \cite{8}. Furthermore, we generalize the problem to non-integer-valued entries.
\end{abstract}

\begin{keywords}
Upper Hessenberg, Bohemian Matrix, Maximum Absolute Determinant.
\end{keywords}
\begin{AMS}
11C20, 15A15, 15B36. 
\end{AMS}

\section{Introduction} \label{intro-sec}
Matrices whose entries are from a small subset of the integers are said to be Bohemian, an abbreviation of \textit{BOunded HEight Matrix of integers}. These matrices appear in many different contexts, including adjacency matrices of graphs \cite{10}, Hadamard matrices \cite{1,2}, random discrete matrices \cite{4} and alternating sign matrices \cite{3}.  They have been studied for over a century and remain a subject of active research.  The website \cite{7} provides a comprehensive overview of recent results and open problems.   

Recently, Chan {\it et al.} investigated several properties of the characteristic polynomials of upper Hessenberg Bohemian matrices \cite{5}, and Thornton {\it et al.} obtained a number of results concerning the distribution of their eigenvalues, characteristic heights, and maximum absolute determinant values \cite{6}. These papers state several conjectures on the values of the determinants of Bohemian matrices \cite{7}, many of which have recently been solved and generalised by Massimiliano Fasi and Gian M. N. Porzio \cite{8}. 
  
One of these conjectures, which is a refinement of a result of Li Ching \cite{9}, was until now lacking a solution that could be generalised. We here provide a generalisable solution for that problem and explore an extension of it.  Specifically, we focus on the maximum absolute determinant of upper Hessenberg matrices with fixed subdiagonal entries and a given population $[0,t]$ of upper triangular entries. Our calculations provide a more comprehensive understanding of the behaviour of the determinants of these kind of matrices and include special cases that had previously been solved by other approaches.  \bigskip

\section{Results}
 In 1993, Ching proved the following theorem.  
\begin{theorem}\label{theorem2.1}
The maximum absolute determinant of an $n \times n$ upper-Hessenberg matrix with upper triangular entries from $\{0,1\}$ and subdiagonal entries fixed at $1$ is given by the Fibonacci sequence. \cite{9} 
\end{theorem}

Recently Fasi and Porzio have established the following theorem, proving a result conjectured by Thornton \cite{7}: 

 \begin{theorem}\label{theorem2.2}
 The maximum absolute determinant of an $n \times n$ upper-Hessenberg matrix with upper triangular entries from $\{0,1,2\}$ and subdiagonal entries fixed at $1$ is given by the following sequence.  Let $M_n$ denote the maximum absolute determinant among these $n \times n$ matrices, then $M_1=2$, $M_2=4$, and $M_n=2\cdot M_{n-1}+M_{n-2}$ for all $n \geq 3$. \cite{8}
 \end{theorem}
  
It is a natural question what happens if the upper triangular population is $\{0,1,\dots,n\}$. Fasi and Porzio stated the following conjecture in this context. 
\begin{conjecture}\label{conjecture2.3}
The maximum absolute determinant of an $n \times n$ upper-Hessenberg matrix with upper triangular entries from $\{0,1,\dots,d\}$ and subdiagonal entries fixed at $1$ is given by the following generalized Fibonacci sequence.  Let $M_n$ denote the maximum absolute determinant among these $n \times n$ matrices, then $M_1=d$, $M_2=d^2$, $M_n=d\cdot M_{n-1}+M_{n-2}$ for all $n \geq 3$. \cite{8}
\end{conjecture}

We discuss the following, further generalisation of this problem.  

\begin{problem}\label{problem2.4}
What is the maximum absolute determinant of an $n \times n$ upper-Hessenberg matrix with upper triangular entries drawn from $[0,t]$, $t > 0$, and subdiagonal entries taking a fixed value $s \in \reals$?
\end{problem}

  If $s$ is negative the problem is relatively straightforward and the result can be stated as the following theorem, whose proof may be found in \cite{8}. 
  
  \begin{theorem}\label{theorem2.5}
  The maximum absolute determinant of an $n \times n$ upper-Hessenberg matrix with upper triangular entries from $[0,t]$ and subdiagonal entries fixed at $s < 0$ is given by $t\cdot (t-s)^{n-1}$ for all $n \in \mathbb{N}$.
  \end{theorem}
  
 However the proof of this theorem does not extend to $s > 0$.  
 
 We here solve Problem \ref{problem2.4} in various regimes when $s > 0$.  The first case we consider is when $s \leq t $, the second is $t\cdot(1+\epsilon) \geq s \geq t$ (where $\epsilon$ is a sufficiently small positive number depending on the dimension of the matrix) and the third case is $s > t \cdot \dfrac{4}{5} \cdot n^2$.  Our main results are contained in the following theorems:  
  \begin{theorem}\label{theorem2.6}
  The maximum absolute determinant of an $n \times n$ upper-Hessenberg matrix with entries from the interval $[0,t]$ and subdiagonal entries fixed at $s$, such that $t \geq s > 0 $, is given by the following sequence.  Let $M_n$ denote the maximum absolute determinant among these $n \times n$ matrices; then $M_1=t$, $M_2=t^2$ and $M_n=tM_{n-1}+s^2M_{n-2}$ for all $n>2$ in $\mathbb{Z}$.
  \end{theorem}
 
 \begin{remark}\label{remark2.7}
 {\rm Note that the determinant is a linear function with respect to each entry. Hence, for the problem of the matrix with maximum absolute determinant, the cases when the upper triangular population is $\{0,1,\dots,d\}$ and $[0,d]$ are equivalent. So, if we set $t=d\in \mathbb{N}$ and $s=1$, then we prove Conjecture \ref{conjecture2.3} as a corollary of Theorem \ref{theorem2.6}.}
 \end{remark}
  
\begin{theorem}\label{theorem2.8}
The maximum absolute determinant of an $n \times n$ upper-Hessenberg matrix with entries from the interval $[0,t]$, subdiagonal entries fixed at $s$ and $n\geq 4$ such that $t\cdot(1+\epsilon(n)) \geq s \geq t$ where $\epsilon(n)$ is a sufficiently small positive number depending on $n$ is given by the following sequence.  Let $M_n$ denote the maximum absolute determinant among these $n \times n$ matrices; then  $M_4=3s^2t^2$, $M_5=s^4t+4s^2t^3$ and $M_n=tM_{n-1}+s^2M_{n-2}$ for all $n>5$ in $\mathbb{Z}$.
\end{theorem}

\begin{theorem}\label{theorem2.9}
The maximum absolute determinant of an $n \times n$ upper-Hessenberg matrix with entries from the interval $[0,t]$ and subdiagonal entries fixed at $s$ such that $\dfrac{s}{t} >  \dfrac{4}{5} \cdot n^2$ is given by $s^{n-1}t+ \Big\lfloor \dfrac{n}{2}\Big\rfloor \Big\lfloor\dfrac{n-1}{2} \Big\rfloor s^{n-3}t^3$.
\end{theorem}
  
 Throughout this paper we use the following definitions and notation. 
   
 \begin{definition}\label{definition2.10}
{\rm The set of all $n \times n$ upper Hessenberg matrices whose subdiagonal entries are fixed at $s$ and upper triangular entries are from the set $P$ is denoted by $\mathcal{G}_s^{n \times n} (P)$.}
 \end{definition}

 \begin{definition}\label{definition2.11}
{\rm For a given $n \times n$ matrix $A$, denote the determinant of the bottom-right $(n+1-k) \times (n+1-k)$ part of $A$ by $H_k(A)$, and for convenience set $H_{n+1}(A) \coloneqq 1$ for any $n \times n$ matrix $A$.}
 \end{definition}

\begin{notation}\label{notation2.12}
{\rm For a given matrix $A$, when referring to the matrix itself we use square brackets and when referring to the determinant of $A$ we use straight brackets. For instance, $A=\begin{bmatrix}1&2\\1&1\end{bmatrix}$, $\det(A)=|A|=\begin{vmatrix}1&2\\1&1\end{vmatrix}=-1$. To avoid confusion, we use $\abs(\cdot)$ for absolute value sometimes.}
\end{notation}
\bigskip
\section{Case I}
Firstly we deal with the case $t \geq s > 0$ of Problem \ref{problem2.4}. To prove Theorem \ref{theorem2.6} we introduce several lemmas. For convenience, set $M_0 \coloneqq 1$. 
\begin{lemma}\label{lemma3.1}
$M_n \geq t \cdot M_{n-1}$, $\forall n \in \mathbb{Z}^+$.
\end{lemma}
\begin{proof}
Suppose that maximum absolute determinant for $(n-1) \times (n-1)$ matrices is attained by a matrix $B_{(n-1)\times(n-1)}$. Then the following inequality holds trivially 
  \begin{align*}
  M_n \geq abs \Bigg(\begin{vmatrix}\begin{array}{c|cccc}
  t&0&\cdots&&0 \\ \hline s&&&& \\ 0&&B&& \\ \vdots&&&& \\ 0&&&&
  \end{array}\end{vmatrix} \Bigg)=t\cdot M_{n-1}.
  \end{align*}
\end{proof}
  
  For the next lemma, we start by writing the determinant for any matrix $A \in \mathcal{G}_s^{n \times n} ([0,t])$ using Laplace expansion twice, firstly for the first column of the matrix $A$ and secondly for the first rows of the resulting two matrices.
\begin{align}
|A|&= \footnotesize \begin{vmatrix} a_{11} & a_{12} & a_{13}& \cdots  & a_{1(n-1)} & a_{1n} \\ s & a_{22} & a_{23} & \cdots & a_{2(n-1)} & a_{2n} \\ 0 & s & a_{33} & \cdots & a_{3(n-1)} & a_{3n} \\  0 & 0 & s & \ddots & a_{4(n-1)} & a_{4n}\\ \vdots & \ddots & \ddots &  \ddots & \ddots & \vdots \\ 0 & 0 & 0 &   \cdots & s & a_{nn} \end{vmatrix} =a_{11}\begin{vmatrix}   a_{22} & a_{23} & \cdots & a_{2(n-1)} & a_{2n} \\  s & a_{33} & \cdots & a_{3(n-1)} & a_{3n} \\   0 & s & \ddots & a_{4(n-1)} & a_{4n}\\ \vdots & \ddots &  \ddots & \ddots & \vdots \\  0 & 0 &   \cdots & s & a_{nn} \end{vmatrix}-s \cdot \begin{vmatrix}   a_{12} & a_{13} & \cdots & a_{1(n-1)} & a_{1n)} \\  s & a_{33} & \cdots & a_{3(n-1)} & a_{3n} \\   0 & s & \ddots & a_{4(n-1)} & a_{4n}\\ \vdots & \ddots &  \ddots & \ddots & \vdots \\  0 & 0 &   \cdots & s & a_{nn} \end{vmatrix}=
   \nonumber\\ &=(a_{11}a_{22}-s\cdot a_{12})\cdot \footnotesize\begin{vmatrix} a_{33} & \cdots & a_{3n-1} & a_{3n} \\  s & \ddots & a_{4(n-1)} & a_{4n}\\ \vdots &  \ddots & \ddots & \vdots \\   0 &   \cdots & s & a_{nn} \end{vmatrix}+(-1)^1(a_{11}a_{23}-s\cdot a_{13}) \cdot \footnotesize\begin{vmatrix} \begin{array}{c|cccc} s & a_{34} & a_{35} & \cdots  & a_{3n} \\  \hline 0 & a_{44} & a_{45} & \cdots  & a_{4n} \\   0 & s & a_{55} & \ddots & a_{5n}\\ \vdots & \ddots &  \ddots & \ddots & \vdots \\  0 & 0 &   \cdots & s & a_{nn} \end{array}   \end{vmatrix}+
\nonumber\\
 & +(-1)^2(a_{11}a_{24}-s\cdot a_{14}) \cdot \footnotesize \begin{vmatrix} \begin{array}{cc|ccc} s & a_{33} & a_{35} & \cdots  & a_{3n} \\ 0 & s & a_{45} & \cdots  & a_{4n} \\  \hline   0 & 0 & a_{55} & \ddots & a_{5n}\\ \vdots & \ddots &  \ddots & \ddots & \vdots \\  0 & 0 &   \cdots & s & a_{nn} \end{array}   \end{vmatrix}+\cdots+ (-1)^{n-2} (a_{11}a_{2n}-s\cdot a_{1n}) \cdot \footnotesize \begin{vmatrix}s& a_{33} &a_{34} & \cdots & a_{3(n-1)}\\ 0 & s & a_{44} & \cdots & a_{4(n-1)} \\ 0& 0 & s & \cdots & a_{5(n-1)} \\ \vdots & \ddots & \ddots & \ddots & \vdots \\ 0 &0 & \cdots & 0 & s \end{vmatrix}=  \hspace{5cm}
\nonumber\end{align}
\begin{align} {\scriptstyle
= \bigg[ (a_{11}a_{22}-s \cdot a_{12})H_3(A)-(a_{11}a_{23}-s \cdot  a_{13})H_4(A) \cdot s\bigg]+\bigg[ (a_{11}a_{24}-s \cdot a_{14})H_5(A)\cdot s^2-(a_{11}a_{25}-s \cdot a_{15})H_6(A)\cdot s^3\bigg]+ \cdots
} 
\nonumber\\ \label{eq:1}
{\scriptstyle \cdots+\bigg[ (a_{11}a_{2(n-1)}-s \cdot a_{1(n-1)})H_n(A) \cdot s^{n-3}-(a_{11}a_{2n}-s \cdot a_{1n})H_{n+1}(A) \cdot s^{n-2}\bigg] }
\end{align}
   (The last term depends on the parity of $n$, if $n$ is even, it is ${\scriptstyle \bigg[(a_{11}a_{2n}-s \cdot a_{1n})H_{n+1}(A) \cdot s^{n-2}\bigg] }$.)

 Now we state a new lemma which is inspired by the previous expansion. 
\begin{lemma}\label{lemma3.2}
For all k in $\{2,3,\dots,n-1\}$ we have the following inequality: \begin{equation}\label{eq:2}
  \bigg| (a_{11}a_{2k}-s \cdot a_{1k}) \cdot H_{k+1}(A)-(a_{11}a_{2(k+1)}-s\cdot a_{1(k+1)})\cdot  s \cdot H_{k+2}(A)  \bigg| \leq t^2 \cdot M_{n-k}+t  s^2 \cdot M_{n-k-1}
  \end{equation}
  and also, for the case when $n$ is even: \begin{equation}\label{eq:3}
  \bigg| (a_{11}a_{2n}- s \cdot  a_{1n}) \cdot H_{n+1}(A) \bigg| \leq  t^2 \cdot M_0
  \end{equation} 
\end{lemma}

\begin{proof}
For \eqref{eq:2}, it is enough to show the following:  \begin{equation}\label{eq:4}
   \smash{\displaystyle \max_{x \in [-st,t^2], y \in [-st^2,s^2t]}}|x \cdot H_{k+1}(A)+y \cdot H_{k+2}(A)| \leq t^2M_{n-k}+ts^2M_{n-k-1}
  \end{equation} 
  
It suffices to check four extreme cases of $x$ and $y$, i.e, $$(x,y)\in \{(-st,s^2t),(-st,-st^2), \\ (t^2,s^2t), (t^2,-st^2) \}$$

Using $t \geq s > 0$, the triangle inequality, the definition of $M_n$, and Lemma \ref{lemma3.1} we get the following inequalities for these four cases:

\begin{itemize}
\item[\textit{1.}] $ (x,y)=(-st,s^2t)$:
\begin{align*}
|(-st)H_{k+1}(A)+(s^2t)H_{k+2}(A)| &\leq (st)|H_{k+1}(A)|+(s^2t)|H_{k+2}(A)| \leq \\ & \leq (st)M_{n-k}+(s^2t)M_{n-k-1} \leq \\ & \leq t^2 M_{n-k}+s^2tM_{n-k-1} \  \checkmark  
\end{align*}

\item[\textit{2.}] $ (x,y)=(-st,-st^2)$:
\begin{align*}
|(-st)H_{k+1}(A)+(-st^2)H_{k+2}(A)| & \leq (st)|H_{k+1}(A)|+(st^2)|H_{k+2}(A)| \leq \\ & \leq (st)M_{n-k}+(st^2)M_{n-k-1}=\\ & = (st)M_{n-k}+ts(t-s)M_{n-k-1}+ts^2M_{n-k-1} \leq \\ & \leq (st)M_{n-k} + s(t-s)M_{n-k} + ts^2M_{n-k-1}= \\ & = (2ts-s^2)M_{n-k} + ts^2M_{n-k-1} \leq \\ & \leq t^2M_{n-k} + ts^2M_{n-k-1}  \ \checkmark
\end{align*}

\item[\textit{3.}] $ (x,y)=(t^2,s^2t)$:
\begin{align*}
|(t^2)H_{k+1}(A)+(s^2t)H_{k+2}(A)| & \leq t^2|H_{k+1}(A)|+s^2t|H_{k+2}(A)| \leq \\ & \leq  t^2M_{n-k}+s^2tM_{n-k-1} \  \checkmark
\end{align*}

\item[\textit{4.}] $ (x,y)=(t^2,-st^2)$:

$$|(t^2)H_{k+1}(A)+(-st^2)H_{k+2}(A)|=$$ 
\begin{align*} 
&= t^2\cdot \abs\Bigg( \begin{vmatrix} a_{(k+1)(k+1)} & a_{(k+1)(k+2)} & \cdots   & a_{(k+1)n} \\ s & a_{(k+2)(k+2)} & \cdots & a_{(k+2)n} \\ \vdots & \ddots  & \ddots & \vdots \\ 0  &   \cdots & s & a_{nn} \end{vmatrix} -s \cdot \begin{vmatrix} a_{(k+2)(k+2)} & a_{(k+2)(k+3)} & \cdots   & a_{(k+2)n} \\ s & a_{(k+3)(k+3)} & \cdots & a_{(k+3)n} \\ \vdots & \ddots  & \ddots & \vdots \\ 0  &   \cdots & s & a_{nn} \end{vmatrix} \Bigg)= \\ & = t^2 \cdot \abs\Bigg( \big(a_{(k+1)(k+1)}-s \big)H_{k+2}(A)- s \cdot \begin{vmatrix} a_{(k+1)(k+2)} & a_{(k+1)(k+3)} & \cdots   & a_{(k+1)n} \\ s & a_{(k+3)(k+3)} & \cdots & a_{(k+3)n} \\ \vdots & \ddots  & \ddots & \vdots \\ 0  &   \cdots & s & a_{nn} \end{vmatrix} \Bigg) \leq  \\ & \leq   t^2 \cdot \abs\Bigg( \big(a_{(k+1)(k+1)}-s\big)H_{k+2}(A)\Bigg)+t^2 s \cdot \abs\Bigg( \begin{vmatrix} a_{(k+1)(k+2)} & a_{(k+1)(k+3)} & \cdots   & a_{(k+1)n} \\ s & a_{(k+3)(k+3)} & \cdots & a_{(k+3)n} \\ \vdots & \ddots  & \ddots & \vdots \\ 0  &   \cdots & s & a_{nn} \end{vmatrix} \Bigg) \leq  
\\
&\leq t^2 \cdot \abs\Bigg( \big(a_{(k+1)(k+1)}-s\big)H_{k+2}(A)\Bigg)+t^2s \cdot M_{n-k-1} = \\ &
 =  t^2 \cdot \big|a_{(k+1)(k+1)}-s\big| \cdot \big|H_{k+2}(A)\big| + t^2s \cdot M_{n-k-1} \leq \\ & \leq  t^2 \cdot \big|a_{(k+1)(k+1)}-s\big| \cdot M_{n-k-1} + t^2s  \cdot M_{n-k-1} \leq  \\ & \leq t^2 \cdot \max \{t-s,s\} \cdot M_{n-k-1} + t^2s  \cdot M_{n-k-1} \hspace{20cm}
\end{align*} 
 \begin{itemize}
 \item[\textit{4.1}] If $t-s= \max \{t-s,s\}$:  \begin{align*}
 |t^2H_{k+1}(A)-st^2H_{k+2}(A)| &\leq  t^2 \cdot (t-s) \cdot M_{n-k-1} + t^2s  \cdot M_{n-k-1}=\\&= t^3M_{n-k-1} \leq \\ & \leq t^2 M_{n-k} \leq \\ & \leq t^2M_{n-k}+ts^2M_{n-k-1} \ \checkmark
 \end{align*}
 \item[\textit{4.2}] If $s= \max \{t-s,s\}$: \begin{align*}
 |t^2H_{k+1}(A)-st^2H_{k+2}(A)| &\leq t^2 \cdot s \cdot M_{n-k-1} + t^2s  \cdot M_{n-k-1}= \\ & = (2t^2s-ts^2)M_{n-k-1}+ts^2M_{n-k-1} \leq \\ & \leq (2ts-s^2)M_{n-k}+ ts^2M_{n-k-1} \leq \\ & \leq t^2M_{n-k}+ ts^2M_{n-k-1} \  \checkmark
 \end{align*}
 
\end{itemize} 

\end{itemize}

Hence \eqref{eq:4} is done. And \eqref{eq:3} is trivial.
\end{proof}

 Using the triangle inequality and Lemma \ref{lemma3.2} in \eqref{eq:1} yields the following inequality for $M_n$: 
\begin{equation} \label{eq:5}
M_n \leq t^2M_{n-2}+ts^2M_{n-3}+t^2s^2M_{n-4}+ts^4M_{n-5}+t^2s^4M_{n-6}+ts^6M_{n-7}+\cdots
\end{equation}
Note as well that we have $M_0=1$, $M_1=t$, $M_2=t^2$. 

 Define a new sequence $(K_n)_{n \geq 0}$ in the following way: $K_0=1$, $K_1=t$, $K_2=t^2$ and $K_n=tK_{n-1}+s^2K_{n-2}$, for all $n \geq 3$. We state two simple lemmas concerning this sequence. 
\begin{lemma}\label{lemma3.3}
$K_n=tK_{n-1}+ts^2K_{n-3}+ts^4K_{n-5}+\cdots$ for all $n \geq 1$.
\end{lemma}

\begin{proof}
We use induction on $n$. For $n=1$ and $n = 2$, the equality is trivial. It is straightforward to verify that if the statement holds for $n\in \{1,2,\dots,k\}$ and then it holds for $n=k+1$ where $k \geq 2$. By using the induction assumption: $$K_{k+1}=tK_{k}+s^2K_{k-1}=tK_{k}+ts^2K_{k-2}+ts^4K_{k-4}+ts^6K_{k-6}+\cdots $$ 
\end{proof}
 
 \begin{lemma}\label{lemma3.4}
 $K_n=t^2K_{n-2}+ts^2K_{n-3}+t^2s^2K_{n-4}+ts^4K_{n-5}+\cdots$ for all $n \geq 2$.
 \end{lemma}

\begin{proof}
For $n=2$ it is clear. Using Lemma \ref{lemma3.3}, for any $n \geq 3$:
\begin{align*}
K_{n}&=tK_{n-1}+s^2K_{n-2}=\\ & =t \cdot \big(tK_{n-2}+ts^2K_{n-4}+ts^4K_{n-6}+\cdots\big) + s^2 \cdot \big(tK_{n-3}+ts^2K_{n-5}+ts^4K_{n-7}+\cdots\big)= \\ & =t^2K_{n-2}+ts^2K_{n-3}+t^2s^2K_{n-4}+ts^4K_{n-5}+\cdots
\end{align*}
\end{proof}

\begin{lemma}\label{lemma3.5}
$K_n \geq M_n$ for all $n \in \mathbb{N}$.
\end{lemma}

\begin{proof}
Notice that we already know, by Lemma \ref{lemma3.4} and \eqref{eq:5}, that
$$K_0=M_0, \ K_1=M_1, \ K_2=M_2$$ $$K_n=t^2K_{n-2}+ts^2K_{n-3}+t^2s^2K_{n-4}+ts^4K_{n-5}+\cdots$$ $$M_n \leq t^2M_{n-2}+ts^2M_{n-3}+t^2s^2M_{n-4}+ts^4M_{n-5}+\cdots$$

By induction, it is clear that $K_n \geq M_n$, $\forall n \in \mathbb{N}$.
\end{proof}

\begin{lemma}\label{lemma3.6}
$K_n \leq M_n$ for all $n \in \mathbb{N}$.
\end{lemma}

\begin{proof}
It suffices to give an example $A \in \mathcal{G}_s^{n \times n} ([0,t])$ for which the absolute determinant value is equal to $K_n$. Define an $n \times n$ matrix as follows: 
\begin{equation}\label{e:7}
 U_n(s,t)= U_n \coloneqq \begin{footnotesize} \begin{bmatrix}t & 0 & t & 0 & \cdots && \cdots  \\ s & t& 0& t& \cdots && \cdots \\ 0 & s & t & 0 & \cdots && \cdots  \\ 0 & 0 & s & t & \ddots && \cdots  \\ \vdots & \ddots & \ddots & \ddots & \ddots &&\vdots \\ 0 & 0& 0& 0& \ddots &t & 0\\0 & 0& 0&0 & \cdots & s & t \end{bmatrix}_{n\times n}\end{footnotesize} 
\end{equation}
 i.e., $a_{ij}=t$ if $j \geq i$ and $j-i$ is even; $a_{ij}=0$ if $j > i$ and $j-i$ is odd. Note that $|U_1|=t=K_1$, $|U_2|=t^2=K_2$ and by using Laplace expansion it is easy to see that 
\begin{equation}\label{e:8}
|U_n|=t \cdot |U_{n-1}|+s^2 |U_{n-2}|
\end{equation} 
  which is the same recurrence relation as for $(K_n)_{n\geq 1}$. Hence, for all positive integers $n$, we have $K_n = |U_n| \leq M_n$. 
\end{proof}

 As a corollary of Lemma \ref{lemma3.5} and \ref{lemma3.6}, $M_n=|U_n|=K_n$. This completes the proof of Theorem \ref{theorem2.6}.  \begin{flushright}  QED.  \end{flushright}
 
  In the next section we discuss what happens when $s \geq t>0$.  We are able to say something in two regimes. \bigskip
\section{Case II}
   We consider first the case when $s \gtrsim t$, i.e., $s$ is slightly greater than $t$. 
  
  Define permutation matrices $P_n^{(r)}$ and $P_n^{(c)}$ in the following way:  Let $P_n^{(r)}$ be obtained by interchanging the first two rows of the $n \times n$ identity matrix and $P_n^{(c)}$ be obtained by interchanging the last two columns of the $n \times n$ identity matrix. And then define $U_n^{(r)}(s,t)=U_n^{(r)}, \ U_n^{(c)}(s,t)=U_n^{(c)}, \ U_n^{(rc)}(s,t)=U_n^{(rc)}$ in $\mathcal{G}_s^{n \times n} (\{0,t\})$ for $n \geq 4$ as follows:

 \begin{align*}
  P_n^{(r)} \cdot U_n^{(r)}=\footnotesize \begin{bmatrix}s & 0 & t & 0 & \cdots && \cdots  \\ t & t& 0& t& \cdots && \cdots \\ 0 & s & t & 0 & \cdots && \cdots  \\ 0 & 0 & s & t & \ddots && \cdots  \\ \vdots & \ddots & \ddots & \ddots & \ddots &&\vdots \\ 0 & 0& 0& 0& \ddots &t & 0\\0 & 0& 0&0 & \cdots & s & t \end{bmatrix}, \ U_n^{(c)} \cdot P_n^{(c)} = \footnotesize \begin{bmatrix}t & 0 & t & 0 & \cdots && \cdots  \\ s & t& 0& t& \cdots && \cdots \\ 0 & s & t & 0 & \cdots && \cdots  \\ 0 & 0 & s & t & \ddots && \cdots  \\ \vdots & \ddots & \ddots & \ddots & \ddots &&\vdots \\ 0 & 0& 0& 0& \ddots &t & 0\\0 & 0& 0&0 & \cdots & t & s \end{bmatrix} 
  \end{align*}
 \begin{equation} \label{e:12}
 \text{and} \  P_n^{(r)} \cdot U_n^{(rc)} \cdot P_n^{(c)} = \begin{footnotesize} \begin{bmatrix}s & 0 & t & 0 & \cdots && \cdots  \\ t & t& 0& t& \cdots && \cdots \\ 0 & s & t & 0 & \cdots && \cdots  \\ 0 & 0 & s & t & \ddots && \cdots  \\ \vdots & \ddots & \ddots & \ddots & \ddots &&\vdots \\ 0 & 0& 0& 0& \ddots &t & 0 \\ 0 & 0& 0&0 & \cdots & t & s \end{bmatrix} \end{footnotesize}
 \end{equation}
   
   We can determine the relation between the determinants of these matrices and the determinant of $U_n$ in the following way.
   
    Using the Laplace expansion for the first column of $P_n^{(r)} \cdot U_n^{(r)}$ we obtain
  \begin{equation}\label{e:13}
 \abs(|U_n^{(r)}|)=\abs\big(\big|P_n^{(r)} \cdot U_n^{(r)}\big|\big)= s\cdot |U_{n-1}| +ts\cdot  |U_{n-2}|
 \end{equation}
and similarly
\begin{equation}\label{e:14}
 \abs(|U_n^{(c)}|)=s\cdot |U_{n-1}| +ts\cdot  |U_{n-2}|
 \end{equation}
 for all $n \geq 4$.  And again using the Laplace expansion for both the first column and the last row of $P_n^{(r)} \cdot U_n^{(rc)} \cdot P_n^{(c)}$ we get
 \begin{align}
|U_n^{(rc)}|&=\big|P_n^{(r)} \cdot U_n^{(rc)} \cdot P_n^{(c)}\big|= s^2\cdot |U_{n-2}| +2ts^2\cdot  |U_{n-3}|+t^2s^2\cdot |U_{n-4}|  \label{e:15}
 \end{align}
 for all $n\geq4$, defining $ |U_{0}| \coloneqq 0$ for convenience.
   
  \begin{proposition}\label{proposition4.1}
  For the statement of Theorem \ref{theorem2.6}, $1$ is the exact upper bound for $\dfrac{s}{t}$. More explicitly, for any $s>t$ and $n \geq 4$, the matrix that gives the maximum absolute determinant is not $U_n$.
  \end{proposition}
\begin{proof}
If $s>t$, by \eqref{e:8} and \eqref{e:15} we have the following when $ n \geq 4$:
 \begin{align}\label{e:16}
 |U_n^{(rc)}|=s^2|U_{n-2}|+2ts^2|U_{n-3}|+t^2s^2|U_{n-4}|&>s^2|U_{n-2}|+ts^2|U_{n-3}|+\big(t^3|U_{n-3}|+t^2s^2|U_{n-4}|\big) = \nonumber \\
 &= s^2|U_{n-2}|+ts^2|U_{n-3}|+t^2|U_{n-2}|= \nonumber \\
 &=s^2|U_{n-2}|+t|U_{n-1}|=|U_n|
 \end{align}
 That means Theorem \ref{theorem2.6} is not valid for any $s>t$ and $n\geq 4$.
\end{proof}  
  \begin{proposition}\label{proposition4.2}
  For any $n\geq4$, $s,t > 0$ consider the matrix in $\mathcal{G}_s^{n \times n} ([0,t])$ which has the maximum absolute determinant among this set.  We call it maximizing matrix. By Theorem \ref{theorem2.6} we know the maximizing matrix for $0< \dfrac{s}{t} \leq 1$.  Proposition \ref{proposition4.1} shows that the maximizing matrix changes at $\dfrac{s}{t}=1$. We claim that the maximizing matrix changes from $U_n$ to $U_n^{(rc)}$ at $\dfrac{s}{t}=1$.
  \end{proposition}
  
\begin{proof}
Suppose that for an $n \geq 4$ the maximizing matrix changes from $U_n$ to a matrix $R_n=R_n(s,t) \in \mathcal{G}_s^{n \times n} ([0,t])$ at $\dfrac{s}{t}=1$. Then the absolute value of the determinant of $R_n(1,1)$ must be equal to the absolute value of the determinant of $U_n(1,1)$. In \cite{9}, Li Ching proved that $\abs(|R_n(1,1)|)=|U_n(1,1)|$ if and only if $R_n(1,1) \in \big\{ U_n(1,1), U_n^{(r)}(1,1), U_n^{(c)}(1,1), U_n^{(rc)}(1,1)\big\}$. Hence $R_n \in \big\{ U_n, U_n^{(r)}, U_n^{(c)}, U_n^{(rc)} \big\}$.
  
   Furthermore, we have the following inequalities for $2t>s>t$ with $n\geq 4$:
\begin{align}\label{e:17}
&|U_{n-3}|(2t-s)(s-t)s+s^2(s-t)^2|U_{n-4}|>0 \nonumber \\
&\Rightarrow 3ts^2|U_{n-3}|+(s^4+t^2s^2)|U_{n-4}| > (s^3+2st^2)|U_{n-3}|+2ts^3|U_{n-4}| \nonumber \\
&\Rightarrow s^2|U_{n-2}|+2ts^2|U_{n-3}|+t^2s^2|U_{n-4}| > s|U_{n-1}|+ts|U_{n-2}| \nonumber \\
&\Rightarrow |U_n^{(rc)}|> \abs(|U_n^{(c)}|)= \abs(|U_n^{(r)}|)
 \end{align} by \eqref{e:8}, \eqref{e:13}, \eqref{e:14} and \eqref{e:15}.

  So, in \eqref{e:16} and \eqref{e:17} we have shown that for $2>\dfrac{s}{t}>1$ the maximum absolute determinant is attained by $U_n^{(rc)}$ among these four matrices. Therefore $R_n=U_n^{(rc)}$.
\end{proof}
  
  In view of the fact that $U_n^{(rc)}$ has become an important matrix in our investigation, we establish the recurrence relation satisfied by its determinant  in the following proposition.
  \begin{proposition}\label{proposition4.3}
  $|U_4^{(rc)}|=3s^2t^2$, $|U_5^{(rc)}|=s^4t+4s^2t^3$ and $\big(|U_n^{(rc)}|\big)_{n\geq 4}$ satisfies the same recurrence relation as $M_n$:
   \begin{equation}
  |U_n^{(rc)}|=t\cdot|U_{n-1}^{(rc)}|+s^2\cdot|U_{n-2}^{(rc)}|
  \end{equation}
  for any $n\geq 6$.
  \end{proposition}
  
\begin{proof}
By substituting \eqref{e:15} and using \eqref{e:8}
  \begin{align}
   |U_n^{(rc)}|&=s^2|U_{n-2}|+2ts^2|U_{n-3}|+t^2s^2|U_{n-4}|= \nonumber\\
   &=t\big(s^2|U_{n-3}|+2ts^2|U_{n-4}|+t^2s^2|U_{n-5}|\big)+s^2\big(s^2|U_{n-4}|+2ts^2|U_{n-5}|+t^2s^2|U_{n-6}|\big)= \nonumber\\
   &=t\cdot|U_{n-1}^{(rc)}|+s^2\cdot|U_{n-2}^{(rc)}| \nonumber
  \end{align}
\end{proof}
  
\begin{remark}\label{remark4.4}
{\rm By Proposition \ref{proposition4.2} we know that $U_n^{(rc)}$ has the maximum absolute determinant value in the case $1\leq \dfrac{s}{t} \leq 1+\epsilon(n)$ for sufficiently small $\epsilon(n)>0$ depending on $n$. So, Propositions \ref{proposition4.2} and \ref{proposition4.3} finish the proof of Theorem \ref{theorem2.8}. }
\end{remark}
\bigskip 
 \section{Case III}   The third case we consider is when  $s\gg t>0$, i.e. $s$ is sufficiently larger than $t$, depending on $n$. 

We start with an important observation which illustrates the reason why we are interested in $s\gg t$ instead of $s \geq t$. 
\begin{observation}
The case when $s \geq t > 0 $ is excessively general to reach a conclusion on $M_n$.  More explicitly, for this case, there is no fixed matrix structure that gives the maximum absolute determinant for all values of $s>t$.
\end{observation}\label{observation5.1}
\begin{proof}
Consider the case when $s=100$ and $t=1$ and $n=6$. By using MATLAB and testing all possible $2^{21}$ cases, it is not difficult to check that the maximum absolute determinant for this case is $10006000000$ and this value is attained by only two matrices:\\
 $$\footnotesize A_1=\begin{bmatrix}1 & 1& 0 & 0 & 0& 1 \\ 100 & 0 & 1 & 1 & 1& 0 \\ 0 & 100 & 1 & 1 & 1& 0 \\ 0 & 0 & 100 & 0 & 0& 1 \\ 0 & 0 & 0 & 100 & 0& 1 \\ 0 & 0 & 0 & 0 & 100& 1\end{bmatrix}, \ A_2=\begin{bmatrix}1 & 1& 1 & 0 & 0& 1 \\ 100 & 0 & 0 & 1 & 1& 0 \\ 0 & 100 & 0 & 1 & 1& 0 \\ 0 & 0 & 100 & 1 & 1& 0 \\ 0 & 0 & 0 & 100 & 0& 1 \\ 0 & 0 & 0 & 0 & 100& 1\end{bmatrix}$$

  If there is a unique matrix that maximizes the absolute determinant value for all $s>t$ in $\mathcal{G}_s^{6 \times 6} (\{0,t\})$, then it must be the matrix $U_6^{(rc)}$ by Proposition \ref{proposition4.2}. Hence $U_6^{(rc)}(100,1)$ must be the same with either $A_1$ or $A_2$, but clearly it is not. So, the maximizing matrix still depends on the ratio $\dfrac{s}{t}$.
\end{proof}

 The next theorem describes the case $s\gg t>0$ in Problem \ref{problem2.4}. 
 
 \begin{theorem}\label{theorem5.2}
 The maximum absolute determinant of an $n \times n$ upper-Hessenberg matrix with entries from the interval $[0,t]$ and subdiagonal entries fixed at $s$ such that $s \gg t>0 $ (i.e. $s$ can be taken sufficiently larger than $t$ for any case) is given by $s^{n-1}t+ \Big\lfloor \dfrac{n}{2}\Big\rfloor \Big\lfloor\dfrac{n-1}{2} \Big\rfloor s^{n-3}t^3$. (We shall go on afterwards to establish a precise lower bound for $\dfrac{s}{t}$ such that this statement holds.)
 \end{theorem}

Define $M_n \coloneqq \displaystyle\max_{A \in \mathcal{G}_s^{n \times n} ([0,t]) } \abs(|A|)$ again. The proof is in two parts; the first is giving an example to show that $M_n \geq s^{n-1}t+ \Big\lfloor \dfrac{n}{2}\Big\rfloor \Big\lfloor\dfrac{n-1}{2} \Big\rfloor s^{n-3}t^3$, and the second showing that this example has the maximum absolute determinant. We start with the first part. Define a matrix $V_n$ for $n \geq 2$ as follows:  
\begin{equation}\label{f:1}
V_n \coloneqq \footnotesize\begin{bmatrix}t&t&t&t&\cdots&t&0\\s&0&0&0&\cdots&0&t\\0&s&0&0&\cdots&0&t\\0&0&s&0&\cdots&0&t\\\vdots&\vdots&\ddots&\ddots&\ddots&\vdots&\vdots\\0&0&0&0&\cdots&0&t\\0&0&0&0&\cdots&s&t\end{bmatrix}_{n \times n}
\end{equation}

\begin{proposition}\label{proposition5.3}
$|V_n|=(-1)^n(n-1)t^2s^{n-2}$ for all $n \geq 2$.
\end{proposition}
 \begin{proof}
 For $n=2$, it is clear. We use induction on $n$, specifically, we assume that the statement is valid for $n=2,3,\dots,k-1$ and show that it then follows for $n=k$. By using Laplace expansion for the last row, and the induction assumption:
\begin{align*}  \footnotesize |V_k |= \begin{vmatrix} t&t&t&t&\cdots&t&0\\  s&0&0&0&\cdots&0&t\\0&s&0&0&\cdots&0&t\\0&0&s&0&\cdots&0&t\\\vdots&\vdots&\ddots&\ddots&\ddots&\vdots&\vdots\\0&0&0&0&\cdots&0&t\\0&0&0&0&\cdots&s&t  \end{vmatrix}_{k \times k} &=  t \cdot \footnotesize \begin{vmatrix} \begin{array} {cccccc|c}t&t&t&t&\cdots&t&t\\ \hline s&0&0&0&\cdots&0&0\\  0&s&0&0&\cdots&0&0\\0&0&s&0&\cdots&0&0\\\vdots&\vdots&\ddots&\ddots&\ddots&\vdots&\vdots\\0&0&0&0&\cdots&0&0\\0&0&0&0&\cdots&s&0 \end{array} \end{vmatrix}_{(k-1) \times (k-1)}-s \cdot |V_{k-1}|=  
  \\ &= t\cdot (ts^{k-2}(-1)^{k-2}) - s \cdot ((-1)^{k-1}(k-2)t^2s^{k-3})= \\ & = t^2s^{k-2}(-1)^{k}+(k-2)t^2s^{k-2}(-1)^{k}= \\ & = (-1)^k(k-1)t^2s^{k-2}
  \end{align*}
  
 \end{proof}
 
 Now, we define a new sequence of matrices $(W_n)_{n \geq 2}$ depending on the parity of $n$.
  
\begin{equation}\label{e:10}
 \footnotesize W_{2k+1}\coloneqq \begin{bmatrix} \overmat{k}{t&t&\cdots&t}& \overmat{k}{0&0&\cdots&0}&t\\s&0&\cdots&0&t&t&\cdots&t&0\\0&s&\cdots&0&t&t&\cdots&t&0\\ \vdots&&\ddots&&\vdots&\vdots&\ddots&\vdots&\vdots \\ 0&0&\cdots&s&t&t&\cdots&t&0\\ 0&0&\cdots&0&s&0&\cdots&0&t\\ 0&0&\cdots&0&0&s&\ddots&0&t\\\vdots&&\ddots&&\ddots&&\ddots&\vdots&\vdots\\0&0&\cdots&&\cdots&&\cdots&s&t\end{bmatrix}  \ \, \ \
  \footnotesize W_{2k+2} \coloneqq \begin{bmatrix} \overmat{k}{t&t&\cdots&t}& \overmat{k+1}{0&0&0&\cdots&0&t} \\s&0&\cdots&0&t&t&t&\cdots&t&0\\0&s&\cdots&0&t&t&t&\cdots&t&0\\ \vdots&&\ddots&&\vdots&\vdots&\ddots&\ddots&\vdots&\vdots \\ 0&0&\cdots&s&t&t&t&\cdots&t&0\\ 0&0&\cdots&0&s&0&0&\cdots&0&t\\ 0&0&\cdots&0&0&s&0&\ddots&0&t\\ \vdots&&\ddots&&\ddots&&\ddots&&\vdots&\vdots\\ \vdots&&\ddots&&\ddots&&\ddots&&\vdots&\vdots\\ 0&0&\cdots&&\cdots&&\cdots&&s&t \end{bmatrix}
\end{equation}  
  
   Notice that $W_{2k+1}$ contains a $k \times k$ block full of $t$'s whereas $W_{2k+2}$ contains a $k \times (k+1)$ block; and note that if we define $W_{2k+2}$ such that it would contain a $(k+1) \times k$ block of $t$'s instead of a $k \times (k+1)$ block it would not change the determinant value. For later convenience we denote the alternative version $W_{2k+2}'$

\begin{proposition}\label{proposition5.4}
$|W_n|=\Big\lfloor\dfrac{n-1}{2} \Big\rfloor(-s)^{n-3}t^3-s|W_{n-1}|$ for all $n \geq 3$.
\end{proposition}
\begin{proof}
We need to show $|W_{2k+1}|=k t^3(-s)^{2k-2}-s|W_{2k}|$ and $|W_{2k+2}|=k t^3(-s)^{2k-1}-s|W_{2k+1}|$. We are going to show only one, because the proofs in both cases are essentially identical.
  
   Using Laplace expansion for the last row of $W_{2k+2}$ (and by iteratively -- exactly $k$ times -- doing it to the last row of each the resulting matrices), and by using Proposition \ref{proposition5.3} we get:
\begin{align*}
|W_{2k+2}|&=t \cdot s^k\cdot (-1)^k \cdot |V_{k+1}|-s|W_{2k+1}|= \\ & = t \cdot s^k\cdot (-1)^k \cdot (-1)^{k+1}kt^2s^{k-1}-s|W_{2k+1}|= \\ & = kt^3(-s)^{2k-1}-s|W_{2k+1}|
\end{align*}

\end{proof}
 
\begin{proposition}\label{proposition5.5}
$|W_n|=(-1)^{n-1} \Big( s^{n-1}t+ \Big\lfloor \dfrac{n}{2}\Big\rfloor \Big\lfloor\dfrac{n-1}{2} \Big\rfloor s^{n-3}t^3 \Big)$ for all $n \geq 3$.
\end{proposition}
 \begin{proof}
 Using the previous proposition, this follows straightforwardly by induction.
 \end{proof}
   
 As a corollary of the last proposition, we have \begin{equation}\label{eq:6}
 M_n \geq \abs(|W_n|)=s^{n-1}t+ \Big\lfloor \dfrac{n}{2}\Big\rfloor \Big\lfloor\dfrac{n-1}{2} \Big\rfloor s^{n-3}t^3
 \end{equation} 
 for all $n \geq 2$.
 
 So the first part of the proof is done. Next, we prove the reverse implication. 
 
  Since we are looking for the maximum absolute determinant it suffices to check $\mathcal{G}_s^{n \times n} (\{0,t\})$ instead of $\mathcal{G}_s^{n \times n} ([0,t])$. 
  
   Notice that the determinant value of any $A \in\mathcal{G}_s^{n \times n} (\{0,t\})$ is a polynomial in $s$ and $t$. More explicitly, the determinant of $A$ is an element of the following set:
$$(-1)^{n-1}s^{n-1}t \cdot \{0,1\}+(-1)^{n-2}s^{n-2}t^2 \cdot \Big\{0,1, \dots , {{n-1}\choose{1}} \Big\}+(-1)^{n-3}s^{n-3}t^3\cdot \Big\{0,\dots , {{n-1}\choose{2}} \Big\}+\cdots+t^n\cdot \{0,1\}$$

 Recall that we assumed $s\gg t$ and already know \eqref{eq:6}. Then
 \begin{align}\label{eq:7}
M_n \geq s^{n-1}t+ & \Big\lfloor \dfrac{n}{2}\Big\rfloor \Big\lfloor\dfrac{n-1}{2} \Big\rfloor s^{n-3}t^3 >s^{n-2}t^2 \cdot {{n-1}\choose{1}}+s^{n-4}t^4 \cdot {{n-1}\choose{3}}+s^{n-6}t^6 \cdot {{n-1}\choose{5}}+\cdots= \nonumber \\ &= \abs \Big( (-1)^{n-2}s^{n-2}t^2 \cdot {{n-1}\choose{1}}+(-1)^{n-4}s^{n-4}t^4 \cdot {{n-1}\choose{3}}+\cdots \Big)
 \end{align}

This means that the sign of the determinant which gives the maximum absolute determinant cannot be $(-1)^{n-2}$. So, its sign is $(-1)^{n-1}$. 

 By \eqref{eq:6} and using $s\gg t$ again,
 \begin{equation}\label{eq:8}
  M_n \geq s^{n-1}t +\Big\lfloor \dfrac{n}{2}\Big\rfloor \Big\lfloor\dfrac{n-1}{2} \Big\rfloor s^{n-3}t^3 > s^{n-3}t^3\cdot {{n-1}\choose{2}}+s^{n-5}t^5\cdot {{n-1}\choose{4}}+ \cdots 
\end{equation}  

 So, the maximum absolute determinant must contain the term $s^{n-1}t$, i.e., the maximum absolute determinant is of the form $s^{n-1}t+\cdots$.

Moreover, we also have the following 
\begin{align} \label{eq:9}
M_n &\geq s^{n-1}t+ \Big\lfloor \dfrac{n}{2}\Big\rfloor \Big\lfloor\dfrac{n-1}{2} \Big\rfloor s^{n-3}t^3 > \nonumber \\  & > s^{n-1}t + (-1) s^{n-2}t^2+s^{n-3}t^3\cdot {{n-1}\choose{2}}+s^{n-5}t^5\cdot {{n-1}\choose{4}}+ \cdots
\end{align}

 This means that the maximum absolute determinant cannot contain the term $s^{n-2}t^2$, i.e., the maximum absolute determinant is of the form $1 \cdot s^{n-1}t+0 \cdot s^{n-2}t^2 +\cdots$. 
 
 Again by $s\gg t$ and \eqref{eq:6}, we can state:  \begin{align} \label{eq:10}
 M_n &\geq s^{n-1}t+ \Big\lfloor \dfrac{n}{2}\Big\rfloor \Big\lfloor\dfrac{n-1}{2} \Big\rfloor s^{n-3}t^3> \nonumber \\ & >s^{n-1}t+ \Big( \Big\lfloor \dfrac{n}{2}\Big\rfloor \Big\lfloor\dfrac{n-1}{2} \Big\rfloor - 1\Big) s^{n-3}t^3+s^{n-5}t^5\cdot {{n-1}\choose{4}}+s^{n-7}t^7\cdot {{n-1}\choose{6}}+ \cdots
 \end{align} 
 
  Therefore, the coefficient of  $s^{n-3}t^3$ in the maximum absolute determinant must be at least $\Big\lfloor \dfrac{n}{2}\Big\rfloor \Big\lfloor\dfrac{n-1}{2} \Big\rfloor$.
  
  \begin{lemma}\label{lemma5.6}
  For any matrix $A \in \mathcal{G}_s^{n \times n} (\{0,t\})$, if the coefficient of $s^{n-2}t^2$ in $|A|$ is zero, then the absolute value of the coefficient of $s^{n-3}t^3$ is at most $\Big\lfloor \dfrac{n}{2}\Big\rfloor \Big\lfloor\dfrac{n-1}{2} \Big\rfloor$.
  \end{lemma}
  
\begin{proof}
Consider the determinant as a polynomial with variable $s$. It is easy to see that for $$\footnotesize A=\begin{bmatrix} a_{11} & a_{12} & a_{13}& \cdots  & a_{1(n-1)} & a_{1n} \\ s & a_{22} & a_{23} & \cdots & a_{2(n-1)} & a_{2n} \\ 0 & s & a_{33} & \cdots & a_{3(n-1)} & a_{3n} \\  0 & 0 & s & \ddots & a_{4(n-1)} & a_{4n}\\ \vdots & \ddots & \ddots &  \ddots & \ddots & \vdots \\ 0 & 0 & 0 &   \cdots & s & a_{nn} \end{bmatrix},$$
  the determinant of $A$ can be expressed as follows:
\begin{align}\label{eq:11}
|A|= & (-1)^{n-1}s^{n-1} \cdot a_{1n} + (-1)^{n-2} s^{n-2} \cdot  \big(a_{11}a_{2n}+a_{12}a_{3n}+a_{13}a_{4n}+\cdots+a_{1(n-1)}a_{nn}\big)+ \nonumber \\ & + (-1)^{n-3} s^{n-3} \cdot \Big( \smash{\displaystyle\sum_{1\leq i <j \leq (n-1)}} {a_{1i}a_{(i+1)j}a_{(j+1)n}} \Big)+ \cdots
\end{align} 
 
 Because the coefficient of $s^{n-2}$ in \eqref{eq:11} is zero by the assumption of the lemma, \begin{equation}\label{eq:12}
 a_{11}a_{2n}+a_{12}a_{3n}+a_{13}a_{4n}+\cdots+a_{1(n-1)}a_{nn}=0
 \end{equation}
 \begin{equation} \label{eq:13}
 \Rightarrow a_{12}a_{3n}+a_{13}a_{4n}+\cdots+a_{1(n-2)}a_{(n-1)n}=0 \ \text{and} \ a_{11}a_{2n}+a_{1(n-1)}a_{nn}=0
 \end{equation}

 Notice that we can express the coefficient of $s^{n-3}$ in \eqref{eq:11} as follows: \begin{align} \label{eq:14}
 \smash{\displaystyle\sum_{1\leq i <j \leq (n-1)}} {a_{1i}a_{(i+1)j}a_{(j+1)n}} =  \nonumber\\ \nonumber\\
=\ a_{11} \cdot a_{22} \cdot a_{3n}+a_{11}  \cdot a_{23} \cdot a_{4n}+a_{11}  \cdot a_{24} \cdot a_{5n}+a_{11}  \cdot a_{25} \cdot a_{6n}  +\cdots&+a_{11}  \cdot a_{2(n-1)} \cdot a_{nn} + \nonumber\\
+a_{12} \cdot a_{33} \cdot a_{4n}+a_{12}  \cdot a_{34} \cdot a_{4n}+a_{12}  \cdot a_{35} \cdot a_{6n}+ \cdots&+a_{12}  \cdot a_{3(n-1)} \cdot a_{nn}+ \nonumber\\ 
+a_{13} \cdot a_{44} \cdot a_{5n}+a_{13}  \cdot a_{45} \cdot a_{6n}+ \cdots & +a_{13}  \cdot a_{4(n-1)} \cdot a_{nn}+\nonumber\\
+a_{14} \cdot a_{55} \cdot a_{6n}+ \cdots &+a_{14}  \cdot a_{5(n-1)} \cdot a_{nn}+ \nonumber\\
 \ddots \hspace{1cm} & \ddots \hspace{1cm} \vdots \nonumber\\
 &+a_{1(n-2)} \cdot a_{(n-1)(n-1)} \cdot a_{nn}
 \end{align}

 The first equality at \eqref{eq:13} means that there are at least $n-3$ zeros among the set $$\{a_{12},a_{13},\dots,a_{1(n-2)}\}\cup\{a_{3n},a_{4n},\dots,a_{(n-1)n} \}$$
 
  Suppose that there are $k_1$ and $k_2$ zeros in $\{a_{12},a_{13},\dots,a_{1(n-2)}\}$ and $\{a_{3n},\dots,a_{(n-1)n} \}$ respectively. Note that we can see the expansion \eqref{eq:14} as an upper triangular half of an $(n-2) \times (n-2)$ chessboard by observing that setting $a_{1i}=0$ corresponds to colouring the $i^{th}$ row (from the top) black, and setting $a_{jn}=0$ corresponds to colouring the $(n-j+1)^{st}$ column (from the right) black for $i \in\{2,3,\dots,n-2\}$ and $j\in\{3,4,\dots,(n-1)\}$. (Notice that we do not colour for the cases $a_{11}=0$ and $a_{nn}=0$) (See Figure 1). 

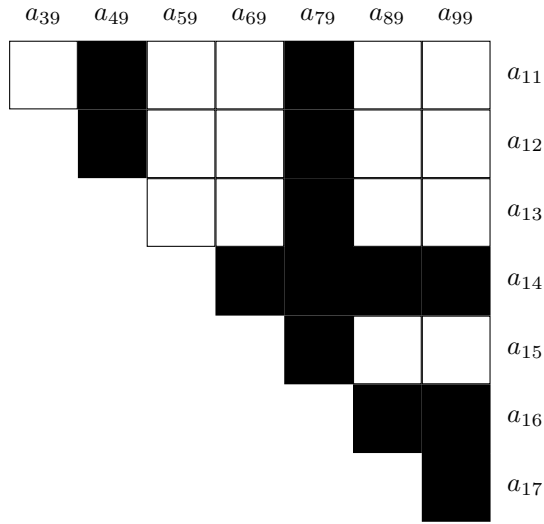
\begin{figure}[H] 
\centering
\begin{tikzpicture}
\matrix (A) [matrix of nodes,
    nodes={draw, fill, minimum size=9mm}]
    {|[fill=white]|~&~&|[fill=white]|~&|[fill=white]|~&~&|[fill=white]|~&|[fill=white]|~\\
     &~&|[fill=white]|~&|[fill=white]|~&~&|[fill=white]|~&|[fill=white]|~\\
     &&|[fill=white]|~&|[fill=white]|~&~&|[fill=white]|~&|[fill=white]|~\\
     &&&~&~&~&~\\
     &&&&~&|[fill=white]|~&|[fill=white]|~\\
     &&&&&~&~\\
     &&&&&&~\\};
    \foreach \i [count=\ni from 3, count=\nj from 1] in {1,...,7}{ 
        \node[right=1mm of A-\i-7.east]{$a_{1\nj}$};
        \node[above=1mm of A-1-\i.north]{$a_{\ni9}$};
    }  
\end{tikzpicture} 
\caption{Case $n=9$ and $a_{14}=a_{16}=a_{17}=a_{49}=a_{79}=0$} 
\end{figure}

  The number of black squares is less than or equal to the number of zero terms in the expansion \eqref{eq:14}. We know that the total number of colored rows and columns is at least $(n-3)$ since $k_1+k_2 \geq n-3$. It is clear that to minimize the number of the black squares, coloured columns must be the leftmost ones and coloured rows must be lowermost ones, and $|k_1-k_2|$ must be $0$ or $1$ (depending on the parity of $n$). See Figure 2 for the minimizing example in the case $n=9$.
  
  It is easy to calculate that the minimum number of black squares is at least \begin{equation} \label{eq:15}
2\cdot \dfrac{n-3}{2}\cdot \dfrac{n-1}{2}\cdot \dfrac{1}{2}=\dfrac{n^2-4n+3}{4}
\end{equation}
 if $n$ is odd, and \begin{equation} \label{eq:16}
 \dfrac{n-4}{2}\cdot \dfrac{n-2}{2}\cdot \dfrac{1}{2}+\dfrac{n-2}{2}\cdot \dfrac{n}{2}\cdot \dfrac{1}{2}=\dfrac{n^2-4n+4}{4}
 \end{equation}
 if $n$ is even.

Hence, in the expansion \eqref{eq:14}, we know that at least $\dfrac{n^2-4n+3}{4}$ or $\dfrac{n^2-4n+4}{4}$ terms are zero. And the number of terms in \eqref{eq:14}, is ${{n-1}\choose{2}}=\dfrac{n^2-3n+2}{2}$.

\begin{figure}[H]
\centering
\begin{tikzpicture}
\matrix (A) [matrix of nodes,
    nodes={draw, fill, minimum size=9mm}]
    {~&~&~&|[fill=white]|~&|[fill=white]|~&|[fill=white]|~&|[fill=white]|~\\
     &~&~&|[fill=white]|~&|[fill=white]|~&|[fill=white]|~&|[fill=white]|~\\
     &&~&|[fill=white]|~&|[fill=white]|~&|[fill=white]|~&|[fill=white]|~\\
     &&&|[fill=white]|~&|[fill=white]|~&|[fill=white]|~&|[fill=white]|~\\
     &&&&~&~&~\\
     &&&&&~&~\\
     &&&&&&~\\};
    \foreach \i [count=\ni from 3, count=\nj from 1] in {1,...,7}{ 
        \node[right=1mm of A-\i-7.east]{$a_{1\nj}$};
        \node[above=1mm of A-1-\i.north]{$a_{\ni9}$};
    }

\end{tikzpicture} 
\caption{Case $n=9$ and $a_{15}=a_{16}=a_{17}=a_{39}=a_{49}=a_{59}=0$} 
\end{figure}
 
 Therefore, by \eqref{eq:15} and \eqref{eq:16} the number of nonzero terms in the \eqref{eq:14} is less than or equal to \begin{equation} \label{r:1}
 \dfrac{n^2-3n+2}{2}-\dfrac{n^2-4n+3}{4}=\dfrac{n^2-2n+1}{4}=\Big\lfloor \dfrac{n}{2}\Big\rfloor \Big\lfloor\dfrac{n-1}{2} \Big\rfloor
\end{equation}    if $n$ is odd, and 
\begin{equation}\label{r:2}
\dfrac{n^2-3n+2}{2}-\dfrac{n^2-4n+4}{4}=\dfrac{n^2-2n}{4}=\Big\lfloor \dfrac{n}{2}\Big\rfloor \Big\lfloor\dfrac{n-1}{2} \Big\rfloor
\end{equation}  if $n$ is even.
  As a consequence, by \eqref{eq:11}, \eqref{r:1} and \eqref{r:2} the absolute value of the coefficient of $s^{n-3}t^3$ is at most $\Big\lfloor \dfrac{n}{2}\Big\rfloor \Big\lfloor\dfrac{n-1}{2} \Big\rfloor$.
\end{proof}
 Hence, the maximum absolute determinant is $$1 \cdot s^{n-1}t+0 \cdot s^{n-2}t^2 +\Big\lfloor \dfrac{n}{2}\Big\rfloor \Big\lfloor\dfrac{n-1}{2} \Big\rfloor s^{n-3}t^3+\cdots$$
 by \eqref{eq:9}, the previous lemma and \eqref{eq:10}.
 
  Now we are going to consider the entries of the matrix that makes the coefficient of $s^{n-3}t^3$ equal to $\Big\lfloor \dfrac{n}{2}\Big\rfloor \Big\lfloor\dfrac{n-1}{2} \Big\rfloor$.
 
 If $n$ is odd, say $n=2k+1$, we have a $(2k-1) \times (2k-1)$ half chessboard, and we colour $2k-2$ rows and columns in total. To make the coefficient $\Big\lfloor \dfrac{n}{2}\Big\rfloor \Big\lfloor\dfrac{n-1}{2} \Big\rfloor$, we colour exactly $(k-1)$ rows (the lowermost ones) and $(k-1)$ columns (the leftmost ones). This means we set \begin{equation}\label{eq:17}
  a_{1(2k-1)}=a_{1(2k-2)}=\cdots=a_{1(k+1)}=0
\end{equation}  
 \begin{equation}\label{eq:18}
a_{3(2k+1)}=a_{4(2k+1)}=\cdots=a_{(k+1)(2k+1)}=0
\end{equation}

 Additionally all the nonzero terms in \eqref{eq:14} must be $t^3$, which means
\begin{equation}\label{eq:19}
a_{11}=a_{12}=\cdots=a_{1k}=t
\end{equation} 
 \begin{equation}\label{eq:20}
a_{(k+2)(2k+1)}=a_{(k+3)(2k+1)}=\cdots=a_{(2k+1)(2k+1)}=t
\end{equation}
 \begin{equation}\label{eq:21}
 a_{ij}=t \  \text{for all} \ (i,j) \in\{2,3,\dots,(k+1)\} \times \{(k+1),(k+2),\dots,(2k)\}
 \end{equation}
because there is a term $a_{1(i-1)} \cdot a_{ij} \cdot a_{(j+1)(2k+1)}$ that corresponds to a white square in the chessboard representation for each $ (i,j) \in\{2,3,\dots,(k+1)\} \times \{(k+1),(k+2),\dots,(2k)\}$.
 
  Moreover, recall the second equation at $\eqref{eq:13}$ and we already have $a_{11}=a_{(2k+1)(2k+1)}=t$ by \eqref{eq:19} and \eqref{eq:20}, then \begin{equation}\label{eq:23}
  a_{2(2k+1)}=a_{1(2k)}=0
  \end{equation}
 
We already know that the coefficient of $s^{n-1}$ in \eqref{eq:11} is not zero, so \begin{equation} \label{r:4}
   a_{1(2k+1)}=t
   \end{equation}
  
  So far we have arrived at the following matrix structure by \eqref{eq:17}-\eqref{r:4}\vspace{2mm} 
   \begin{equation}\label{f:2}
   \begin{bmatrix} \overmat{k}{t&t&\cdots&t}& \overmat{k}{0&0&\cdots&0}&t\\s&?&?&?&t&t&\cdots&t&0\\0&s&?&?&t&t&\cdots&t&0\\ \vdots&&\ddots&?&\vdots&\vdots&\ddots&\vdots&\vdots \\ 0&0&\cdots&s&t&t&\cdots&t&0\\ 0&0&\cdots&0&s&?&?&?&t\\ 0&0&\cdots&0&0&s&?&?&t\\\vdots&&\ddots&&\ddots&&\ddots&?&\vdots\\0&0&\cdots&&\cdots&&\cdots&s&t\end{bmatrix}
   \end{equation}

  For the case when $n$ is even, we have the same, but according to the choice of the difference between the number of zero terms in the sets $\{a_{12},a_{13},\dots,a_{1(n-2)}\}$ and $\{a_{3n},\dots,a_{(n-1)n} \}$, i.e. $(k_1-k_2)$, as $-1$ or $+1$; we are going to have $W_{2k+2}$ or $W_{2k+2}'$ as defined at $\eqref{e:10}$. From now on, we just consider the odd case, the even case can be done in exactly the same way.

 Recall that we know that the maximum absolute determinant is $$ 1 \cdot s^{n-1}t+0 \cdot s^{n-2}t^2 +\Big\lfloor \dfrac{n}{2}\Big\rfloor \Big\lfloor\dfrac{n-1}{2} \Big\rfloor s^{n-3}t^3+\cdots$$
 
  Because of \eqref{eq:6} and $s\gg t$ we have the following inequality, 
   $$M_n \geq s^{n-1}t+ \Big\lfloor \dfrac{n}{2}\Big\rfloor \Big\lfloor\dfrac{n-1}{2} \Big\rfloor s^{n-3}t^3>$$ \begin{equation} \label{eq:24}
   >s^{n-1}t+ \Big\lfloor \dfrac{n}{2}\Big\rfloor \Big\lfloor\dfrac{n-1}{2} \Big\rfloor s^{n-3}t^3+(-1) \cdot s^{n-4}t^4+s^{n-5}t^5\cdot {{n-1}\choose{4}}+s^{n-7}t^7\cdot {{n-1}\choose{6}}+ \cdots
   \end{equation}
   
    Hence, the coefficient of $s^{n-4}t^4$ must be $0$ to have the maximum absolute determinant. Now we are going to show that this fact forces all entries with a question mark in \eqref{f:2} to be filled with $0$. 
    
    \begin{proposition}\label{proposition5.7}
    Recall that we have defined $V_n$ in \eqref{f:1}. If any of the entries of the triangular block of $0$'s in the upper half is $t$ instead of $0$, then the determinant contains the term $ s^{n-3}t^3$.
    \end{proposition}
   \begin{proof}
   Consider the following permutation: 
$$\footnotesize\begin{bmatrix}t&\cdots&&\boxed{t}&\cdots&&\cdots&t&0 \\ \boxed{s}&&&\vdots &&&&&t\\ &\boxed{\ddots}&&\vdots&&&&&\\&&\boxed{\ddots}&\vdots&&&&&\vdots\\&&&s&\cdots&\boxed{t}&&&\\&&&&\boxed{\ddots}&\vdots&&&\vdots\\&&&&&s&\cdots&\cdots&\boxed{t}\\&&&&&&\boxed{\ddots}&&\vdots\\&&&&&&&\boxed{s}&t\end{bmatrix}$$

 This permutation gives the term $s^{n-3}t^3$, and because all the terms with $s^{n-3}$ have the same sign, $s^{n-3}t^3$ does not vanish.
   \end{proof}
 \begin{proposition}\label{proposition5.8}
 Consider the matrix in \eqref{f:2}, if there is $t$ in any of the entries that are filled with a question mark, then the determinant has the term $s^{n-4}t^4$.
 \end{proposition}
 \begin{proof}
 Let $n=2k+1$ be odd, the other case can be done by the same way. Suppose that the $t$ is placed in the top-left triangular block of $?$'s WLOG. Then consider the determinant as follows:
  
 $$\scriptsize \begin{vmatrix} \begin{array} {ccccc|cccc} t&t&\cdots&t& 0&0&\cdots&0&t\\s&?&?&?&t&t&\cdots&t&0\\0&s&?&?&t&t&\cdots&t&0\\ \vdots&&\ddots&?&\vdots&\vdots&\ddots&\vdots&\vdots \\ 0&0&\cdots&s&t&t&\cdots&t&0\\ \hline 0&0&\cdots&0&s&?&?&?&\boxed{t}\\ 0&0&\cdots&0&0&\boxed{s}&?&?&t\\\vdots&&\ddots&&\ddots&&\boxed{\ddots}&?&\vdots\\0&0&\cdots&&\cdots&&\cdots&\boxed{s}&t\end{array} \end{vmatrix}$$
 
 The determinant of the upper-left $(k+1) \times (k+1)$ square contains the term $t^3s^{k-2}$ by Proposition \ref{proposition5.7}. And as we have boxed, there is a permutation that gives a $ts^{k-1}$ term from the bottom-right square. When we consider these two together we get the term $t^4s^{2k-3}$ which is exactly what we are looking for.
\end{proof}  
   
  As a corollary, we can state that all the entries with question mark must be filled with $0$ to have the maximum absolute determinant. Therefore the matrix which has the maximum absolute determinant is the matrix $W_{2k+1}$ as we have defined at \eqref{e:10}. \hspace{9.1cm}  QED.  
     
      Now we can discuss the condition $s\gg t$. We have used this in our proof in the following ways in \eqref{eq:7}, \eqref{eq:8}, \eqref{eq:9}, \eqref{eq:10} and \eqref{eq:24}. We can rewrite these inequalities setting $x\coloneqq \dfrac{s}{t}$:
   \begin{align}
   \label{eq:25 }&x^{n-1}+ \Big\lfloor \dfrac{n}{2}\Big\rfloor \Big\lfloor\dfrac{n-1}{2} \Big\rfloor x^{n-3} >x^{n-2} \cdot {{n-1}\choose{1}}+x^{n-4} \cdot {{n-1}\choose{3}}+x^{n-6} \cdot {{n-1}\choose{5}}+\cdots \\
  \label{eq:26} & x^{n-1}+ \Big\lfloor \dfrac{n}{2}\Big\rfloor \Big\lfloor\dfrac{n-1}{2} \Big\rfloor x^{n-3} > x^{n-3}\cdot {{n-1}\choose{2}}+x^{n-5}\cdot {{n-1}\choose{4}}+ \cdots \\
 \label{eq:27}  & x^{n-2}+ \Big\lfloor \dfrac{n}{2}\Big\rfloor \Big\lfloor\dfrac{n-1}{2} \Big\rfloor x^{n-3} > x^{n-3}\cdot {{n-1}\choose{2}}+x^{n-5}\cdot {{n-1}\choose{4}}+ \cdots \\
  \label{eq:28} &x^{n-3}>x^{n-5}\cdot {{n-1}\choose{4}}+x^{n-7}\cdot {{n-1}\choose{6}}+ \cdots \\
 \label{eq:29}  & x^{n-4}>x^{n-5}\cdot {{n-1}\choose{4}}+x^{n-7}\cdot {{n-1}\choose{6}}+ \cdots
   \end{align}
      
   Note that except in the last inequality, these allow us to take $x$ asymptotic to $n^2$.  However, the last one requires $n^4$. Fortunately we can overcome this issue using another way to tackle the problem. Recall that we deduced \eqref{eq:29} from \eqref{eq:24}. Before stating \eqref{eq:24}, we found that the maximum absolute determinant has the form $ 1 \cdot s^{n-1}t+0 \cdot s^{n-2}t^2 +\Big\lfloor \dfrac{n}{2}\Big\rfloor \Big\lfloor\dfrac{n-1}{2} \Big\rfloor s^{n-3}t^3+\cdots$ and the matrix with this determinant has the form \eqref{f:2} 
\begin{lemma}\label{lemma5.9}
Let $$ A = \scriptsize\begin{bmatrix} t&t&\cdots&t& 0&0&\cdots&0&t\\s&a_{22}&\cdots&a_{2k}&t&t&\cdots&t&0\\0&s&\ddots&\vdots&t&t&\cdots&t&0\\ \vdots&&\ddots&a_{kk}&\vdots&\vdots&\ddots&\vdots&\vdots \\ 0&0&\cdots&s&t&t&\cdots&t&0\\ 0&0&\cdots&0&s&a_{(k+2)(k+2)}&\cdots&a_{(k+2)(2k)}&t\\ 0&0&\cdots&0&0&s&\ddots&\vdots&t\\\vdots&&\ddots&&\ddots&\ddots&\ddots&a_{(2k)(2k)}&\vdots\\0&0&\cdots&&\cdots&\cdots&0&s&t\end{bmatrix},$$
where $A \in \mathcal{G}_s^{n \times n} (\{0,t\})$. And let \begin{equation} \label{eq:30}
\abs(|A|)=1 \cdot s^{n-1}t+\Big\lfloor \dfrac{n}{2}\Big\rfloor \Big\lfloor\dfrac{n-1}{2} \Big\rfloor s^{n-3}t^3-b_4 \cdot s^{n-4}t^4+b_5 \cdot s^{n-5}t^5-b_6 \cdot s^{n-6}t^6+\cdots
\end{equation}
Then, $b_m \cdot n \geq b_{m+1}$ for all $m\in\{4,5,\dots,n-1\}$.
\end{lemma}
\vspace{-3mm}
\hspace{5mm}	\textit{Proof.} Recall that in the permutation definition, if we consider the determinant as a function of $s$, the absolute value of the coefficient of $s^{n-l}$ is: \begin{equation} \label{eq:31}
 \smash{\displaystyle\sum_{1\leq i_1<\cdots<i_{l-1} \leq (n-1)}} {a_{1i_1}a_{(i_1+1)i_2}a_{(i_2+1)i_3}\cdots a_{(i_{l-1}+1)n}} 
  \end{equation}
 
 Then \begin{equation} \label{eq:32}
 b_m \cdot t^m=\smash{\displaystyle\sum_{1\leq i_1<\cdots<i_{m-1} \leq (n-1)}} {a_{1i_1}a_{(i_1+1)i_2}a_{(i_2+1)i_3}\cdots a_{(i_{m-1}+1)n}}
 \end{equation}
  \\ 
  
   and 
\begin{equation}\label{eq:33}
b_{m+1} \cdot t^{m+1}=\smash{\displaystyle\sum_{1\leq i_1<\cdots<i_m \leq (n-1)}} {a_{1i_1}a_{(i_1+1)i_2}a_{(i_2+1)i_3}\cdots a_{(i_m+1)n}}
\end{equation}
  
  Clearly some of the terms are going to vanish such as the ones starting with $a_{1(k+1)}$ because $a_{1(k+1)}$ is already determined as $0$. Now we are going to define a function from non-vanishing terms in  \eqref{eq:33} to non-vanishing terms in \eqref{eq:32}. (From now on we consider them not as numbers, but as sequences of $a_{ij}$'s) Define the function $$\phi:  B_{m+1} \coloneqq \Big\{a_{1i_1}a_{(i_1+1)i_2}a_{(i_2+1)i_3}\cdots a_{(i_{m}+1)n} \neq 0 \big| 1 \leq i_1<\cdots<i_{m} \leq (n-1)\Big\} \rightarrow $$ $$
~\ \hspace{2.5cm}\rightarrow B_m  \coloneqq \Big\{a_{1i_1}a_{(i_1+1)i_2}a_{(i_2+1)i_3}\cdots a_{(i_{m-1}+1)n} \neq 0 \big| 1 \leq i_1<\cdots<i_{m-1} \leq (n-1)\Big\} $$ as follows:
\begin{itemize}
\item If $i_2 \leq k (=\frac{n-1}{2})$, we have $a_{1i_2} \neq 0$, $$\phi \big(a_{1i_1}a_{(i_1+1)i_2}a_{(i_2+1)i_3}\cdots a_{(i_{m}+1)n} \big) \coloneqq a_{1i_2}a_{(i_2+1)i_3}a_{(i_3+1)i_4}\cdots a_{(i_{m}+1)n}$$
\item If $i_2 \geq k+1$, then $i_{m-1} \geq i_2 \geq k+1$, we have $a_{(i_{m-1}+1)n} \neq 0$,  $$\phi \big(a_{1i_1}a_{(i_1+1)i_2}\cdots a_{(i_{m-1}+1)i_m}a_{(i_{m}+1)n} \big) \coloneqq a_{1i_1}a_{(i_1+1)i_2}\cdots a_{(i_{m-2}+1)i_{m-1}}a_{(i_{m-1}+1)n}$$
\end{itemize}
   
   We are going to show that the preimage of any element in the range has cardinality less than or equal to $n$.
   
   Consider the preimage of an arbitrary element in the range of $\phi$, $$\phi ^{-1} \Bigg[ \phi \big(a_{1i_1}a_{(i_1+1)i_2}\cdots a_{(i_{m-1}+1)i_m}a_{(i_{m}+1)n} \big) \Bigg] $$
   
   Define a set $P_1 \coloneqq \emptyset$ if $i_2 \geq k+1$, otherwise $P_1 \coloneqq$\\
\scalebox{1}{ $ \coloneqq \Big\{a_{11}a_{2i_2}a_{(i_2+1)i_3}\cdots a_{(i_{m-1}+1)i_m}a_{(i_{m}+1)n}, a_{12}a_{3i_2}a_{(i_2+1)i_3}\cdots a_{(i_{m-1}+1)i_m}a_{(i_{m}+1)n}, \dots  $ }\\
\scalebox{1}{ \hspace{8cm}$\dots, a_{1(i_2-1)}a_{i_2i_2}a_{(i_2+1)i_3}\cdots a_{(i_{m-1}+1)i_m}a_{(i_{m}+1)n} \Big\} $ }

Similarly define $P_2 \coloneqq \emptyset$ if $i_{m-1} \leq k$, otherwise $P_2 \coloneqq$\\
\scalebox{0.8}{ $ \coloneqq \Big\{a_{1i_1}a_{(i_1+1)i_2}\cdots a_{(i_{m-2}+1)i_{m-1}}a_{(i_{m-1}+1)(i_{m-1}+1)}a_{(i_{m-1}+2)n}, a_{1i_1}a_{(i_1+1)i_2}\cdots a_{(i_{m-2}+1)i_{m-1}}a_{(i_{m-1}+1)(i_{m-1}+2)}a_{(i_{m-1}+3)n}, \dots $ }

\scalebox{0.8}{ \hspace{10cm}$\dots, a_{1i_1}a_{(i_1+1)i_2}a_{(i_2+1)i_3}\cdots a_{(i_{m-2}+1)i_{m-1}}a_{(i_{m-1}+1)(n-1)}a_{nn} \Big\} $ }

It is not difficult to see that  $$\phi ^{-1} \Bigg[ \phi \big(a_{1i_1}a_{(i_1+1)i_2}\cdots a_{(i_{m-1}+1)i_m}a_{(i_{m}+1)n} \big) \Bigg] \subseteq P_1 \cup P_2$$ and for the cardinality of $P_1\cup P_2$ we have $$|P_1 \cup P_2| \leq |P_1| + |P_2| \leq (k-1) + (k-1) \leq n $$

Therefore, we get that the preimage of any element in $B_m$ has cardinality less than or equal to $n$ and this means $|B_m|\cdot n \geq |B_{m+1}|$.

 Note that:
 $$b_m \cdot t^m=\smash{\displaystyle\sum_{1\leq i_1<\cdots<i_{m-1} \leq (n-1)}} {a_{1i_1}a_{(i_1+1)i_2}a_{(i_2+1)i_3}\cdots a_{(i_{m-1}+1)n}}=t^m\cdot |B_m| \Rightarrow b_m=|B_m|$$
  and similarly $$b_{m+1}=|B_{m+1}|$$
 Hence \begin{equation*}
\hspace{5.6cm} b_{m+1}= |B_{m+1}| \leq n \cdot |B_m|=n\cdot b_m  \hspace{5.4cm} \cvd
  \end{equation*} 
  
  As a corollary of Lemma \ref{lemma5.9}, we can write the following inequality for \eqref{eq:30} if $\dfrac{s}{t} \geq n$:
\begin{align} \label{eq:34}
\abs(|A|)&=1 \cdot s^{n-1}t+\Big\lfloor \dfrac{n}{2}\Big\rfloor \Big\lfloor\dfrac{n-1}{2} \Big\rfloor s^{n-3}t^3-b_4 \cdot s^{n-4}t^4+b_5 \cdot s^{n-5}t^5-b_6 \cdot s^{n-6}t^6+\cdots= \nonumber\\ &
= s^{n-1}t+\Big\lfloor \dfrac{n}{2}\Big\rfloor \Big\lfloor\dfrac{n-1}{2} \Big\rfloor s^{n-3}t^3-s^{n-5}t^4\big(s\cdot b_4-t \cdot b_5 \big)-s^{n-7}t^6\big(s\cdot b_6-t \cdot b_7 \big)+ \cdots \leq \nonumber\\ &
\leq s^{n-1}t+\Big\lfloor \dfrac{n}{2}\Big\rfloor \Big\lfloor\dfrac{n-1}{2} \Big\rfloor s^{n-3}t^3-s^{n-5}t^4\big(s\cdot b_4-t \cdot n\cdot b_4 \big)-s^{n-7}t^6\big(s\cdot b_6-t \cdot n \cdot b_6 \big)+ \cdots = \nonumber\\ &
= s^{n-1}t+\Big\lfloor \dfrac{n}{2}\Big\rfloor \Big\lfloor\dfrac{n-1}{2} \Big\rfloor s^{n-3}t^3-s^{n-5}t^5b_4 \big(\dfrac{s}{t} -  n  \big)-s^{n-7}t^7b_6\big(\dfrac{s}{t} - n  \big)+ \cdots \leq \nonumber\\ & \leq  s^{n-1}t+\Big\lfloor \dfrac{n}{2}\Big\rfloor \Big\lfloor\dfrac{n-1}{2} \Big\rfloor s^{n-3}t^3
\end{align}

 So having $\dfrac{s}{t} \geq n$ is sufficient for the last case, namely \eqref{eq:29} is not a requirement anymore if $\dfrac{s}{t} \geq n$. Then there are four inequalities left that we still need to deal with: \eqref{eq:25 }-\eqref{eq:28}. 
 
 \begin{lemma}\label{lemma5.10}
 If $\dfrac{s}{t}=x> \dfrac{1}{\cosh^{-1}(2)}\cdot n^2$ these inequalities hold for $n \geq 2$.
 \end{lemma}
\begin{proof}
Start with the first one, \eqref{eq:25 }, we know that $$\dfrac{3}{2} \cdot n< \dfrac{1}{\cosh^{-1}(2)} \cdot n^2<x \ ,$$   then
  \begin{align*}
  x^{n-2}{{n-1}\choose{1}}+x^{n-4}{{n-1}\choose{3}}+x^{n-6}{{n-1}\choose{5}}+\cdots \ & < \ x^{n-2}\cdot n+x^{n-4}\cdot \dfrac{n^3}{3!}+x^{n-6}\cdot\dfrac{n^5}{5!}+\cdots \ < \\ &
< \ \dfrac{2}{3}\cdot x^{n-1}+\big(\dfrac{2}{3}\big)^3 \cdot \dfrac{x^{n-1}}{3!}+ \big(\dfrac{2}{3}\big)^5 \cdot \dfrac{x^{n-1}}{5!}+\cdots \ < \\ & < \ \dfrac{2}{3} \cdot x^{n-1} \cdot \big( 1+ \dfrac{1}{3!} + \dfrac{1}{5!}+ \cdots \big) \ < \\ & < \ \dfrac{2}{3} \cdot x^{n-1} \cdot \sinh(1) \ <  \ x^{n-1} \ < \\ & < \  x^{n-1}+ \Big\lfloor \dfrac{n}{2}\Big\rfloor \Big\lfloor\dfrac{n-1}{2} \Big\rfloor x^{n-3} \ \checkmark
  \end{align*}

For the second one \eqref{eq:26}, use the first inequality \eqref{eq:25 },
\begin{align*}
x^{n-1}+ \Big\lfloor \dfrac{n}{2}\Big\rfloor \Big\lfloor\dfrac{n-1}{2} \Big\rfloor x^{n-3} &> x^{n-2}{{n-1}\choose{1}}+x^{n-4}{{n-1}\choose{3}}+x^{n-6}{{n-1}\choose{5}}+\cdots \  > \\& > \ x^{n-3}\cdot n \cdot {{n-1}\choose{1}}+x^{n-5} \cdot n \cdot {{n-1}\choose{3}}+\cdots > \\ & > x^{n-3}\cdot {{n-1}\choose{2}}+x^{n-5}\cdot {{n-1}\choose{4}}+ \cdots  \ \checkmark 
\end{align*}

For the third case \eqref{eq:27},
\begin{align*}
x^{n-3}\cdot {{n-1}\choose{2}}+x^{n-5}\cdot {{n-1}\choose{4}}+ \cdots \ < \ x^{n-3} \cdot \dfrac{n^2}{2!}+ x^{n-5} \cdot \dfrac{n^4}{4!}+x^{n-7} \cdot \dfrac{n^6}{6!}+ \cdots \ & < \\ < \ 
x^{n-2}  \cdot \dfrac{(\cosh^{-1}(2))^1}{2!}+x^{n-3} \cdot  \dfrac{(\cosh^{-1}(2))^2}{4!}+x^{n-4} \cdot  \dfrac{(\cosh^{-1}(2))^3}{6!}+ \cdots \ &< \\
< \ x^{n-2}  \cdot \dfrac{(\cosh^{-1}(2))^2}{2!}+x^{n-3} \cdot  \dfrac{(\cosh^{-1}(2))^4}{4!}+x^{n-4} \cdot  \dfrac{(\cosh^{-1}(2))^6}{6!}+ \cdots \ &< \\
< \ x^{n-2}  \cdot \dfrac{(\cosh^{-1}(2))^2}{2!}+x^{n-2} \cdot  \dfrac{(\cosh^{-1}(2))^4}{4!}+x^{n-2} \cdot  \dfrac{(\cosh^{-1}(2))^6}{6!}+ \cdots \ &< \\
< \ x^{n-2} \cdot \bigg[\cosh\big(\cosh^{-1}(2) \big) -1 \bigg] \ = \ x^{n-2} \ < \ x^{n-2}+ \Big\lfloor \dfrac{n}{2}\Big\rfloor \Big\lfloor\dfrac{n-1}{2} \Big\rfloor x^{n-3} \ & \checkmark
\end{align*}

And for the last inequality \eqref{eq:28},
\begin{align*}
x^{n-5}\cdot {{n-1}\choose{4}}+x^{n-7}\cdot {{n-1}\choose{6}}+ \cdots \  & < \ x^{n-5} \cdot \dfrac{n^4}{4!} + x^{n-7}\cdot   \dfrac{n^6}{6!}+\cdots \ < \\ &
< x^{n-3}  \cdot \dfrac{(\cosh^{-1}(2))^2}{4!} +  x^{n-4}  \cdot \dfrac{(\cosh^{-1}(2))^3}{6!}+ \cdots \ < \\ &
<  x^{n-3}  \cdot  \dfrac{(\cosh^{-1}(2))^4}{4!} +  x^{n-4}  \cdot \dfrac{(\cosh^{-1}(2))^6}{6!}+ \cdots \ < \\ &
<  x^{n-3}  \cdot  \dfrac{(\cosh^{-1}(2))^4}{4!} +  x^{n-3}  \cdot \dfrac{(\cosh^{-1}(2))^6}{6!}+ \cdots \ < \\ &
< x^{n-3}  \cdot \bigg[\cosh\big(\cosh^{-1}(2) \big) -1 \bigg] \ = \ x^{n-3} \ \checkmark
\end{align*}
\end{proof}
 
\begin{remark}\label{remark5.11}
{\rm For simplicity we can write $\dfrac{4}{5}$ instead of $\dfrac{1}{\cosh^{-1}(2)}$ since we are not interested in the strict lower bound. Thus Theorem \ref{theorem5.2} and Lemma \ref{lemma5.10} finish the proof of Theorem \ref{theorem2.9}. }
\end{remark}
\begin{remark}\label{remark5.12}
{\rm The lower bound in Theorem \ref{theorem2.9} is not sharp, but it is clear from the inequality \eqref{eq:28} that it is asymptotic to $n^2$ for our proof to hold. }
\end{remark}
\bigskip
\section{Concluding Remarks}   We found the maximum absolute determinants (and the matrices giving the corresponding values) for the cases $\dfrac{s}{t}>\dfrac{4}{5} \cdot n^2$ (Theorem \ref{theorem2.9}) and $0 \leq \dfrac{s}{t} \leq 1$ (Theorem \ref{theorem2.6}); in addition we already know the case $\dfrac{s}{t} \leq 0$ (Theorem \ref{theorem2.5}). Furthermore we showed that $1$ is the exact upper bound of $\dfrac{s}{t}$ for Theorem \ref{theorem2.6} (Proposition \ref{proposition4.1}), and found the first maximizing matrix after $\dfrac{s}{t}=1$ (Proposition \ref{proposition4.2}) with the recurrence relation satisfied by its absolute determinant value (Proposition \ref{proposition4.3}).  Nevertheless, for the most part the problem as to what happens if $1<\dfrac{s}{t}<\dfrac{4}{5} \cdot n^2$ still remains open. Note that as $\dfrac{s}{t}$ goes to $1$ from $\dfrac{4}{5} \cdot n^2$, the matrix which gives the maximum absolute determinant changes from having a localized to a much more uniform structure. In view of these observations, further investigations might be based on the following questions:
 \begin{itemize}
 \item[1.] Which matrices maximize the absolute determinant as $\dfrac{s}{t}$ goes to $1$ from $n^2$, in other words, what are the transition forms? 
 \item[2.] It is reasonable to work on the case $s=2$ and $t=1$ as a first step to understanding the transition form. Is $U_n^{(rc)}$ the maximizing matrix in this case when $n \geq 4$?
 \item[3.] How many times do transitions occur depending on the size of the matrices?
 \item[4.] We know that for $n\geq4$ the first transition occurs at $\dfrac{s}{t}=1$ from $U_n$ to $U_n^{(rc)}$. When does the second transition occur, in other words what is the maximum possible value of $\epsilon(n)$ in Remark \ref{remark4.4}? 
 \item[5.] What should be the precise lower bound in Theorem \ref{theorem5.2}? Is it asymptotic to $n^2$?
 \item[6.] When do other transitions occur?
 \item[7.] Note that the number of $t$'s in the matrices $U_n$ \eqref{e:7}, $U_n^{(rc)}$ \eqref{e:12} and $W_n$ \eqref{e:10} are equal. Is it true that for any fixed $n$, all maximizing matrices have the same number of $t$'s? 
 \item[8.] Moreover, there is no alternating sign among the nonzero permutations in the determinants for any of these three matrices. Is this condition valid for all maximizing matrices?

\end{itemize}  
 \bigskip
 
\section*{Acknowledgements}

We thank Professor N. J. Higham for drawing our attention to questions concerning Bohemian Matrices.  The work reported here was carried out as part of an undergraduate summer research project by Mr Ahmet Abdullah Kele\c{s} at the University of Bristol, June--July 2019.  Mr Kele\c{s} is grateful for financial support and the hospitality of the School of Mathematics at the University of Bristol during his visit.  JPK was supported by a Royal Society Wolfson Research Merit Award, EPSRC Programme Grant EP/K034383/1 LMF: $L$-Functions and Modular Forms, and by ERC Advanced Grant 740900 (LogCorRM).


\end{document}